\documentclass[11pt, fleqn]{article}
\usepackage{algcompatible}
\usepackage{algorithm}

\usepackage{graphicx}
\usepackage{xcolor}
\usepackage{mathtools}

\usepackage{amsmath,amsthm,amsfonts}
\usepackage{hyperref}
\usepackage{cleveref}
\usepackage[utf8]{inputenc}
\usepackage{tikz}
\usepackage{float}
\usepackage{thmtools} 

\newtheorem{theorem}{Theorem}

\newtheorem{claim}{Claim}

\newtheorem{definition}{Definition}
\newtheorem{lemma}{Lemma}

\newcommand{\opt}{\textsc{opt}}
\newcommand{\eps}{\epsilon}

\newcommand{\E}{\mathbb{E}}
\newcommand{\R}{{\mathbb{R}}}

\allowdisplaybreaks

\title{Online Load and Graph Balancing for Random Order Inputs
}
\author{Sungjin Im\thanks{ UC Merced, USA \texttt{sim3@ucmerced.edu}. }  
\and Ravi Kumar\thanks{Google, Mountain View USA \texttt{ravi.k53@gmail.com}.}
\and Shi Li\thanks{Nanjing University, China \texttt{shili@nju.edu.cn}.}
\and Aditya Petety\thanks{UC Merced, USA \texttt{apetety@ucmerced.edu}.} 
\and Manish Purohit\thanks{Google, Mountain View USA \texttt{mpurohit@google.com}.}
}

\begin{document}


\maketitle
\begin{abstract}
Online load balancing for heterogeneous machines aims to minimize the makespan (maximum machine workload) by scheduling arriving jobs with varying sizes on different machines. In the adversarial  setting, where an adversary chooses not only the collection of job sizes but also their arrival order, the problem is well-understood and the optimal competitive ratio is known to be $\Theta(\log m)$ where $m$ is the number of machines. In the more realistic random arrival order model, the understanding is limited. Previously, the best lower bound on the competitive ratio was only $\Omega(\log \log m)$.

We significantly improve this bound by showing an $\Omega( \sqrt {\log m})$ lower bound, even for the restricted case where each job has a unit size on two machines and infinite size on the others.  On the positive side, we propose an $O(\log m/\log \log m)$-competitive algorithm, demonstrating that better performance is possible in the random arrival model.
\end{abstract}
\section{Introduction}

Online load balancing is a fundamental problem encountered in parallel/distributed computing and network communication, appearing in various forms. It focuses on effectively distributing limited resources to handle tasks that arrive over time (online). Due to its  importance, this problem has been extensively studied in both practice and theory
\cite{azar2005line,jiang2015survey,ghomi2017load}. 

The model of unrelated machines has received significant attention due to its effective representation of heterogeneous processing capabilities. In this setting, the processing time of a job can differ based on the assigned machine. Minimizing the makespan, the maximum workload across all machines, is one of the most common objectives in this realm.

In many practical scenarios, the sequence of job arrivals is unpredictable, necessitating online load balancing strategies. For minimizing makespan in this online setting, an $O(\log m)$-competitive algorithm is known \cite{azar1995competitiveness,aspnes1997line}, where $m$ represents the number of machines. This algorithm guarantees a solution within a logarithmic factor of the optimal offline solution, which has full knowledge of the entire job sequence in advance. Notably, this competitive ratio is known to be tight under adversarial job arrival orders \cite{azar1995competitiveness}, where the algorithm has no information about future jobs.

One way to circumvent the strong lower bounds is to consider the random arrival order model. In this model, the adversary first defines the machines and jobs, then the jobs arrive in a random permutation. The appeal of this model is its simplicity: no assumptions are made about the overall input, and only the order of arrival is randomized. Crucially, the optimal offline solution remains unchanged. The random arrival model has been exceptionally effective in breaking many lower bound barriers, as seen in the secretary problem, packing integer problem, online facility location problem, and many others. For a comprehensive overview of the results in the random arrival model, see~\cite{gupta_singla_2021}.

Despite the success of the random arrival model in other contexts, research on online load balancing within this framework remains surprisingly thin. To the best of our knowledge, the only non-trivial result for this problem is a lower bound of $\Omega(\log \log m)$ \cite{plank2017online}. 

\subsection{Our Results}

In this paper, we make significant strides in online load balancing for unrelated machines with the random arrival model by establishing the following new bounds. 

\paragraph{Lower Bounds.}  We show lower bounds for the special case of graph balancing \cite{ebenlendr2014graph}.  Here, edges representing jobs arrive online, and we must immediately orient each arriving edge towards one of its endpoints. The objective is to minimize the maximum in-degree of any vertex, i.e., the highest number of incoming edges for a single vertex. This problem is equivalent to the online load balancing scenario where each job has a unit size on exactly two machines and an infinite size on the rest.

Even in this restricted setting, we establish a new lower bound of $\Omega(\sqrt{\log m})$ for any randomized algorithm, where $m$ denotes the number of machines (or the number of vertices in the graph balancing problem). This is an exponential improvement over the previously known lower bound of $\Omega(\log \log m)$. Interestingly, our bound holds even when the instance is a tree and the algorithm knows the structure a priori. This result is presented in Section~\ref{sec:lb-greedy}. 

We further demonstrate that the intuitive greedy algorithm, which assigns each job to the less loaded machine (breaking ties randomly), has a competitive ratio of $\Omega(\log m / \log \log m)$. Expressed in graph balancing terminology, this algorithm orients each edge towards the endpoint with the lower in-degree. This highlights the necessity of deviating from the straightforward greedy approach to achieve a competitive ratio 
substantially better than the $\Theta(\log m)$ bound, which is the best possible in the adversarial order setting. This result can be found in Section~\ref{sec:lb-general}.

\paragraph{Upper Bounds.} On the positive side, we present an algorithm for online load balancing in the random arrival order model that achieves a competitive ratio of $O(\log m / \log \log m)$. Notably, this result applies to the general setting of unrelated machines. While the improvement over the $O(\log m)$ bound in the adversarial model is modest, it nevertheless establishes a meaningful separation between the two models for this problem. This result is in Section~\ref{sec:upper}.

\paragraph{Remark.}
    In the online makespan minimization for unrelated machines problem, we use $m$ to denote the number of machines and $n$ the number of jobs. This is the standard notation used in the scheduling literature. However, in the graph balancing problem, $m$ is the number of edges and $n$ is the number of nodes. Since nodes correspond to machines, a lower bound $f(n)$ for the graph balancing implies a lower bound $f(m)$ for  makespan minimization. In fact,  this differentiation is not critical as all lower bound instances we use for the graph balancing are trees; thus, we have $n = m+1$.

\subsection{Our Techniques}

As a warm-up, we review an instance that demonstrates a lower bound of $\Omega(\log m)$ on the competitive ratio of any deterministic algorithm for adversarial arrivals. Consider a setting with machines indexed from 1 to $m$, where $m$ is assumed to be a power of 2 for simplicity. Jobs arrive in rounds, with each job having a size of 1 and being assignable to only two specific machines. We can represent this constraint by representing jobs as pairs of machine indices. For example, a pair (1, 2) indicates that the job can be assigned to either machine 1 or machine 2.

The instance unfolds as follows. In the first round, jobs with pairs $(1, 2), (3, 4), \ldots, (m-1, m)$ arrive. Without loss of generality, assume the algorithm assigns these jobs to the even-indexed machines. The subsequently arriving jobs can only be assigned to the even-indexed machines, which all already have one job. By focusing on the even-indexed machines and halving their indices, we effectively reduce the problem to $m/2$ machines. Each recursion of this process increases the load on the machines considered in the current round by 1. 
Therefore in the $(\log m)$th round, the algorithm's makespan reaches $\log m$, while an optimal solution has a makespan of 1, establishing the $\Omega(\log m)$ lower bound.

This instance also can be viewed as a tree. See Figure~\ref{fig:lb-classic}. 

\begin{figure}
    \centering
    \includegraphics[width=0.4\linewidth]{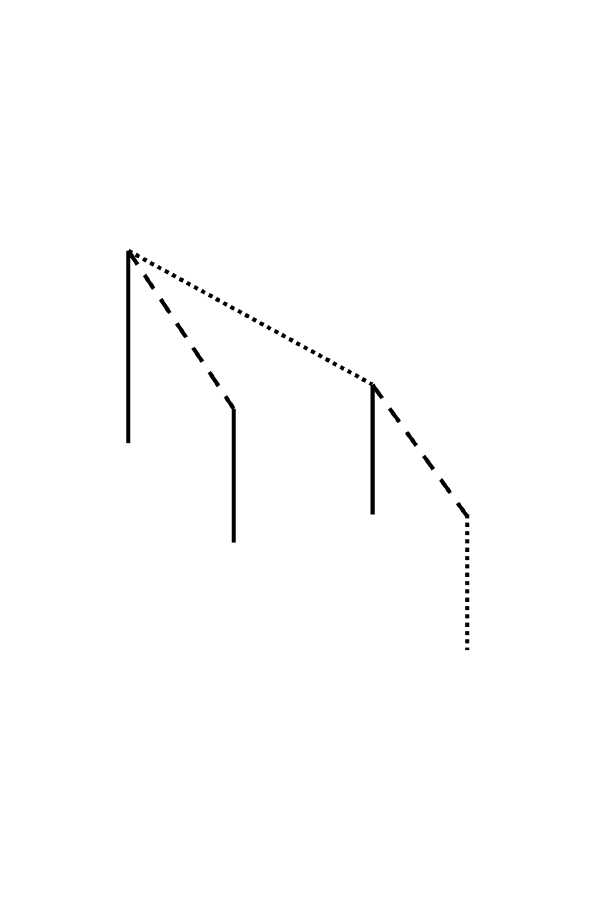}
    \caption{Illustration of a lower bound instance giving $\Omega(\log m)$ lower bound in the adversarial arrival model. Edges arrive in the order of solid, dashed, and dotted, and they are wlog assumed to be oriented towards the root.}
    \label{fig:lb-classic}
\end{figure}
\paragraph{Lower Bounds.} 
In the random arrival order model, orchestrated job sequences that intentionally concentrate load on specific machines are no longer viable, as we cannot dictate the order in which jobs appear. While the aforementioned instance cannot establish a lower bound for arbitrary algorithms within this model, a minor tweak of it enables us to demonstrate a lower bound of $\Omega(\log m/\log \log m)$  for the intuitive greedy algorithm, which assigns each job to the machine with the lesser current load.

To `reproduce' the adversarial order even in the random arrival order model, we use a `fat' tree where every node has a polylogarithmic degree except the leaf nodes; here the tree has an $\Omega(\log m / \log \log m)$ height. The key idea is that in the fat tree, a majority of leaf edges (jobs) arrive before the other edges and we can show that this bad event can occur recursively with a high probability, resulting a high load on the root node.

Unfortunately, we cannot use the same fat tree to show a strong lower bound for arbitrary algorithms. This is because of the following reason. To obtain a good load balancing, it is important to orient edges towards the leaf nodes. However, it is easy for an algorithm to distinguish the leaf nodes from the others early in the random order model because the non-leaf nodes have high degrees.

To overcome this issue, our recursive construction is done more carefully. 
The main property the tree has is that when an edge arrives online, no algorithm should be able to distinguish between the two end points as to which is the parent node for a large number of edges. However, this lower bound instant construction comes at a cost. The tree become significantly fatter and has a height $\Theta(\sqrt{\log m})$, which translates to our lower bound.

\paragraph{Upper Bound.}  For the upper bound on makespan minimization for unrelated machines, we draw inspiration for the algorithm from \cite{molinaro2017online}. The idea is to keep track of a potential function over all machines and reset these potentials after half of the inputs have arrived \cite{gupta2014experts}. The potential function is essentially the softmax function of the machine loads. The algorithm then assigns an arriving job to the machine that incurs the least increase in the potential. The algorithm is identical in both phases so it suffices to analyze only one half. This restarting helps as it ensures sufficient randomness for jobs that arrive later, which is crucial for the analysis. 

\subsection{Other Related Work}

As mentioned before, online load balancing finds numerous applications in parallel/distributed computing and network communication. Thus, a vast literature exists on this topic \cite{azar2005line,jiang2015survey,ghomi2017load,buchbinder2006fair,karger2004simple} and we only cover the most related work. 

The online load balancing problem was introduced by Graham in his seminal work \cite{Graham66,Graham69}. The work showed the list scheduling algorithm that assign an arriving job to the least loaded machine is $(2 - 1/m)$-competitive when all machines are identical. The current best known competitive ratio for this problem is 1.916 \cite{Albers99}. When the machines are uniformly related, meaning that machines have different speeds, $O(1)$-competitive algorithms are known \cite{berman2000line}.
For unrelated machines, $O(q)$-competitive algorithms are known for minimizing the $q$-norm of machine loads \cite{awerbuch1995load,caragiannis2008better}. For extension to multidimensional load balancing, see \cite{azar2013tight,meyerson2013online,ImKKP19}.

There has been a load of work on beyond-worst-case analysis in the recent years. The random arrival order model is a prime example of this framework. Here, we present a selection of recent advances within this model. For the online load balancing problem on identical machines, when jobs arrive in random order, there is a 1.8478-competitive algorithm~\cite{albers2021scheduling} (which is better than known lower bounds in the adversarial setting).
When the elements are revealed in a random order, the online set cover admits a competitive ratio of $O(\log mn)$~\cite{9719768}, which matches the offline bound of $O(\log n)$ when $m$ is polynomial in $n$; here $m$ is the number of sets and $n$ is the number of elements. The online facility location problem  admits a competitive ratio of 3  in the random order model~\cite{kaplan2023almost}. As mentioned, the reader is referred to the survey \cite{gupta_singla_2021} for work in the random arrival order.

Molinaro \cite{molinaro2017online} provides algorithms that are simultaneously good both in the adversarial arrival model and in the random arrival model. In particular, for the makespan objective, the work gives an algorithm that is $\Theta(\log m / \eps)$-competitive in the adversarial order, while simultaneously giving a makespan of $(1+\eps)\opt + O(m \log m / \eps^2)$ in random order. Unfortunately, the analysis requires $\eps \in (0, 1]$ and therefore does not give a competitive ratio better than $O(\log m)$. 

Online load balancing has been recently studied in the presence of ML predictions. The algorithms are given certain compact predictions on the input and can achieve competitive ratios significantly better than $O(\log m)$ when the predictions are almost accurate while asymptotically retaining the worst case guarantees \cite{LattanziLMV20,li2021online}. 

We now discuss related work on the offline load balancing problems. Offline, the load balancing problem for unrelated machines does not admit a better than 1.5-approximation unless P = NP and a 2-approximation is known for the problem due to the seminal work by Lenstra et al.~\cite{lenstra1990approximation}. For the graph balancing problem, when $G$ is known offline, \cite{ebenlendr2008graph} gives a 1.75-approximation.

\newcommand{\makespan}{\mathrm{makespan}}

\section{Preliminaries and Notation}
\label{sec:Prelims}
We consider the online makespan minimization problem on unrelated machines. Let $\mathcal{J}$ be a set of $n$ jobs and let $\mathcal{M}$ be a set of $m$ machines. Let $p_{i,j} \in \mathbb{R}^+ \cup \{\infty\}$ denote the processing load of job $j$ on machine $i$; in the online setting the values of $\{p_{i,j}\}$ for a job $j$ become known to the algorithm when job $j$ arrives.  For an assignment 
$\sigma: \mathcal{J} \rightarrow \mathcal{M}$ of jobs to machines, we define its \emph{makespan} as 
$\makespan(\sigma) = \max_{i \in \mathcal{M}} \sum_{j \in \sigma^{-1}(i)} p_{i,j}$.  The goal of an online  algorithm is to find an assignment $\sigma$ that minimizes $\makespan(\sigma)$.

In this paper we consider the online setting with uniformly random job arrivals. Let $\mathcal{A}$ be an online algorithm. Let $\pi$ be a uniformly random permutation of jobs in $\mathcal{J}$. At each time step $t \in [n]$, a job $\pi(t) \in \mathcal{J}$ arrives and  $\mathcal{A}$ needs to assign it to exactly one machine in $\mathcal{M}$; let $\makespan(\mathcal{A}; \{p_{i,j}\}, \pi)$ be the random variable denoting its makespan. The number of jobs, $n$, is known to the algorithm before any of the jobs arrive. 

Let $\sigma^*$ be the assignment of an optimal, offline algorithm. For any algorithm $\mathcal{A}$, we define its competitive ratio as the worst-case ratio of the expected makespan incurred by the algorithm to the optimal makespan, i.e.,  
\[\text{competitive-ratio}(\mathcal{A}) = \max_{\{p_{i,j}\}} \dfrac{\mathbb{E}_{\pi}[\makespan(\mathcal{A}; \{p_{i,j}\}, \pi)]}{\makespan(\sigma^*)}.\]
Here, we will parameterize the competitive ratio by the number of machines, $m$. So, the competitive ratio is defined over the instances that have at most $m$ machines. 

We also consider a very special case of the makespan minimization problem called \emph{graph balancing}. Let $G = (V,E)$ be an undirected graph. The goal is to orient the edges of the graph $G$ so as to minimize $\max_{v \in V} \text{indegree(v)}$. It can be readily seen that this is a special case of makespan minimization on unrelated machines by considering each vertex as a machine and an edge as a job where each job has a unit load on exactly two machines and an infinite load otherwise.  Again, in this paper, we study the graph balancing problem with edges arriving uniformly at random.

\section{Unrelated Machines: Upper Bound}
    \label{sec:upper}

In this section we give an improved upper bound on the competitive ratio for the online makespan minimization on unrelated machines, under uniformly random job arrivals. We assume that $\opt \geq 1$. We further assume wlog that 
$p_{i,j} \in [0, 1]  \cup \{\infty\}$. This is because we can guess the optimum objective within a constant factor by the standard doubling trick: if the guess turns out to be wrong (we can solve the offline problem within a factor of 2 \cite{lenstra1990approximation}), then we double our guess and pretend that all machines have zero load. When the current guess is $g \geq 2 \opt$, $j$ cannot be assigned to machine $i$ if $p_{i,j} \geq g$; therefore, we can replace the value of $p_{i,j}$ with $\infty$. For all other $p_{i,j} \leq g$, by scaling them down by $g$, we have the assumption. 

\subsection{An $O(\frac{\log m}{\log \log m})$ Upper Bound}

\subsubsection{Algorithm}

The algorithm has two phases that are identical~\cite{gupta2014experts,molinaro2017online}. In the second phase, which starts when the $(\lfloor n/2 \rfloor +1)$st job arrives, it pretends that no jobs have arrived so far. The algorithm uses a softmax of machine loads as the potential, and assigns a new job to the machine that increases the potential the least.

To describe the algorithm, it will be convenient to define load vector $x \in \R_{\geq 0}^m$, which describes the current load of each machine, i.e., if $J_i$ is the set of jobs currently assigned to machine $i$, $x_i = \sum_{j \in J_i} p_{i,j}$. For easy indexing, say job $j^t$ arrives at time $t$. If we assign $j^t$ to machine $i$, it adds to the current load vector a vector $w^t$, whose $i$th entry is $p_{i,j} \in [0,1]$ and all other entries are 0. Notice that $w^t$ has only one non-zero coordinate. Let $s^t$ be the load vector right after the algorithm assigned $j^t$, i.e., $s^t = w^1 + \cdots + w^t$.

We are now ready to define our potential function:

$$\psi(x) = \frac1a\ln\sum_{i \in \mathcal{M}}e^{ax_i},$$
where $x \in \R_{\geq 0}^m$ is the load vector and $a > 0$ is a parameter to be decided later. As mentioned above, a new job is assigned to the machine that increases the potential the least.

\newcommand{\alg}{\textsc{ALG}}

As observed in \cite{gupta2014experts,molinaro2017online}, it suffices to only consider the first half of the jobs. This is because the first half and the second half have the same distribution, and the algorithm ``resets" after seeing the first half of the jobs. This will increase the makespan only by a factor of two. Henceforth, we will only consider the first $n/2$ jobs. 

Define $v^t = \nabla \psi(s^{t-1})$. Notice that $v^t_i = \frac1a\frac{a\cdot e^{ax_i}}{\sum_{i'} e^{ax_{i'}}} = \frac{e^{ax_i}}{\sum_{i'} e^{ax_{i'}}}$. The following is immediate by straightforward calculations. 

\begin{claim}
    \label{claim:deri-change}
    For all $t \geq 1$, we have $||v^t||_1 = 1$ and $v^t \leq e^a v^{t-1}$, where the inequality holds for each coordinate. 
\end{claim}

Using the convexity of the $\psi$ function, we have 
\begin{align*}
    \psi(s^t) - \psi(s^{t-1}) \leq \langle w^t, \nabla\psi(s^t)\rangle  =  \langle w^t, v^{t+1}\rangle \leq e^a \langle w^t, v^t\rangle.
\end{align*}
Then, adding the inequalities over all $t$'s in $\{1, 2, \ldots, n/ 2\}$,  we have 
\begin{align*}
    \psi(s^{n/2}) - \psi(0) \leq e^a \sum_{t = 1}^{n/2} \langle w^t, v^t\rangle.
\end{align*}
Noticing that $\psi(s^{n/2}) \geq ||s^{n/2}||_\infty$ and $\psi(0) = \frac{\ln m}{a}$, we have 
\begin{align*}
    ||s^{n/2}||_\infty \leq e^a \sum_{t = 1}^{n/2} \langle w^t, v^t\rangle + \frac{\ln m}{a}.
\end{align*}

Let $o^t \in [0, 1]^m$ be the load vector induced by the job $j^t$ in the optimum solution. Then $o^t_{i'} = p_{ij_t}$ if $i' = i$ and $0$ otherwise. Here, we are considering the fixed optimum solution that assigns all the $n$ jobs in $J$ to machines. Although each job is assigned to a fixed machine in the optimum solution, $o^t$ is stochastic as $j^t$ is the $t$th job in the random order.

\begin{claim}
    For all $t \geq 1$, we have $\langle v^t, w^t\rangle \leq e^{a}\langle v^t, o^t \rangle$.
\end{claim}
\begin{proof}
\begin{align*}
    \langle w^t, v^{t} \rangle 
    & = \langle w^t, \nabla \psi(s^{t-1}) \rangle \\
    &\leq \psi(s^t) - \psi(s^{t-1}) && \triangleright \mbox{$s^t = s^{t-1} + w^t$ and $\psi$ is convex} \\
    &\leq \psi(s^{t-1}+ o^t) - \psi(s^{t-1}) && \triangleright \mbox{Algorithm is greedy}\\
    &\leq \langle o^t, \nabla \psi(s^{t-1}+o^t)\rangle  && \triangleright \mbox{$\psi$ is convex} \\
    &\leq \langle o^t, e^a \nabla \psi(s^{t-1})\rangle && \triangleright \mbox{Since $||o^t||_\infty \leq 1$} \\
    &= e^a \langle o^t, v^t \rangle.
    \qedhere
\end{align*}    
\end{proof}

Thanks to this claim, we have

\begin{align} \label{sec:restart}
    ||s^{n/2}||_\infty \leq e^{2a} \sum_{t = 1}^{n/2} \langle v^t, o^t\rangle + \frac{\ln m}{a}.
\end{align}

Consider an arbitrary time step $t$. So far, only jobs $j^1, \ldots, j^{t-1}$ have been revealed and they completely determine $v^t$.
Then, $j^t$ is sampled from the other jobs in $J$ uniformly at random. Note that $o^{t} + o^{t+1} + \dots + o^n \leq \opt \cdot \mathbf{1}$. Therefore, we have $\E[o^{t} \mid v^{t}] \leq \frac{\opt}{n-(t-1)} \cdot  \mathbf{1}$. Moreover, $||v^t||_1 \leq 1$; see Claim~\ref{claim:deri-change}. So, we have $\E[\langle v^t, o^t \rangle \mid v^t] \leq \frac{\opt}{n-(t-1)}$. Taking expectation over Eq.~(\ref{sec:restart}), we get 
\begin{align*}
    \E[||s^{n/2}||_{\infty}] & \leq e^{2a}\opt \left(\frac{1}{n} + \frac{1}{n-1} + \cdots + \frac{1}{n/2} \right) + \frac{\ln m}{a} \\
    & \leq 2 e^{2a}\opt + \frac{\ln m}{a}.
\end{align*}
The last inequality above is why we run the algorithm in two phases (otherwise, the summation leads to a $\Theta(\log n)$ term). 

By setting $a = \frac{\ln \ln m}{6}$, we have 
\begin{align*}
       2 e^{2a}\opt + \frac{\ln m}{a} &=  2(\ln m)^{\frac{1}{3}}\opt + \frac{6 \ln m}{\ln \ln m} \\
       &= O \left(\frac{\ln m}{\ln \ln m} \right)  \opt, 
    \end{align*}
where the last equation follows since  $\opt \geq 1$ and $2(\ln m)^{\frac{1}{3}} < \frac{6 \ln m}{\ln \ln m} \text{for $m \geq 2$}$. Thus, we have shown that the maximum load is $O(\frac{\log m}{\log \log m})  \opt$ in expectation for the first half of the arriving jobs.

As discussed, our algorithm restarts after assigning the first half jobs, and the makespan for the second half of jobs can be upper bounded symmetrically. Thus, we have the following.

\begin{theorem}
    There exists an $O(\log m / \log \log m)$-competitive algorithm for minimizing makespan on unrelated machines, when the jobs arrive in uniformly random order. 
\end{theorem}

\section{Graph Balancing: Lower Bounds}
\label{sec:greedy}

In this section we obtain new lower bounds for the online graph balancing problem, which is a special case of unrelated machines, in the random edge arrival model.  Interestingly, our lower bounds hold for tree instances even when the structure of the tree is fully known to the algorithm a priori. It can be easily observed that in any tree instance, an optimal offline solution always attains a maximum load of 1 by simply rooting the tree at an arbitrary node and orienting every edge away from the root. 
We first present our stronger lower bound, which only holds for the greedy algorithm, and then present a lower bound that holds for any algorithm.

\subsection{Lower Bound for a Greedy Algorithm}
    \label{sec:lb-greedy}

Consider the following natural greedy algorithm.  When an edge $e = (u,v)$ arrives, the algorithm (\textsc{Greedy}) assigns $e$ to the endpoint that has the least load at that time, breaking ties uniformly at random. Here, a node's load is the number of already-arrived edges that have been oriented towards it. \textsc{Greedy} is known to have the optimal competitive ratio of $\Theta(\log n)$ for online graph balancing in the adversarial edge arrival model.

In this section we show a lower bound of $\Omega(\log n / \log \log n)$ for \textsc{Greedy} in the random arrival model even when the input graph is a tree. 
The lower bound instance is the following complete tree.

\paragraph{Lower Bound Instance.} Let $T$ be a rooted tree with root $r$ and depth $k$. Each internal node of $T$  has $d = k^{4}$ children. The total number of nodes in $T$ is thus $n = \Theta(d^k) = \Theta(k^{4k})$; so we have $k = \Theta(\log n/\log \log n)$. We assume that $k$ is sufficiently large. 
\smallskip

Before we discuss the intuition and the analysis, we first set up some notation and terminology.  
For ease of analysis, we pretend that edges adjacent to $r$ also have a parent edge (or equivalently, the tree is rooted at a degree 1 node $r'$ that connects to $r$). An edge's \emph{depth} is defined as the number of edges including itself on the path from the root to the edge. Thus, the leaf edges have depth $k$ and the edges incident to the root have depth $1$. A node $u$'s \emph{height} is defined as the number of edges between $u$ to a closest leaf node. Thus, a leaf edge has two end points of height $1$ and $0$ respectively. Similarly, an edge's height is defined as the height of the node of its end point that is closer to the root.

\subsubsection{Intuition} Suppose edges of $T$ arrive in a specific order, from bottom to top---so all leaf edges appear first, and the edges of height two arrive, and so forth. 
Consider the $d$ leaf edges that share a parent node. With high probability, one of these $d$ edges will be oriented towards the parent (and all the others arriving afterwards will be directed towards the leaves). So after all the leaf edges arrive, (almost) all nodes of height 1 will have a load of 1. Then when all the edges of the upper level---equivalently the edges of height 2---arrive, a similar argument shows that two of the edges that share a parent will be directed towards the parent, so nodes of height 2 get a load of 2. Continuing this way, the root node will have a load of $k = \Omega(\log n / \log \log n)$. The primary intuition is that by making the degree of each internal node large enough---even when edges arrive in random order---there is a subtree of a similar structure where edges arrive in this specific leaf-to-root order.

\subsubsection{Analysis}

For the sake of analysis, we let each edge $e$ be associated with a random number $r_e$ sampled from $[0, 1]$ uniformly at random. We assume that edge $e$ arrives at time $r_e$. Note that this way we can simulate the random arrival order of the edges.
As discussed before, the greedy algorithm suffers the most when edges arrive in the bottom-to-top order. Thus, we will seek to show that there is a $k^2$-ary subtree where edges arrive in the bottom-to-top order. To make the analysis more convenient, we will only look at trees where edges arrive in a more structured manner. Let $I_1, \ldots I_k$ be an equal partition of $(0, 1)$ into intervals, i.e., $I_1 = (0, 1/k), I_2 = (1/ k , 2/k), \ldots$ and so on.\footnote{Here, we do not distinguish between closed intervals and open intervals as almost surely no edge will arrive at a time that is a multiple of $1/ k$.}

\begin{definition}[Bad node]
    We say that a node $u$ of height $h$ is \emph{bad} if at least $k^2$ of its child edges arrive during $I_h$; let $B_u$ denote this bad event (for the algorithm).
\end{definition}

The following lemma is a straightforward consequence of the Chernoff bound. We bound the probability that the bad event does not occur, i.e., for a node $u$ of height $h$, at most $k^2$ of its child edges arrive during $I_h$.

\begin{lemma}
\label{lem:badnode}
    For a non-leaf node $u$, $\Pr[ B_u] \geq 1 - \frac{1}{e^{k^3/8}}$.   
\end{lemma}
\begin{proof}
Consider a non-leaf node $u$ with height $h$. We want to calculate the probability that at most $k^2$ child edges arrive in the interval $I_{h}$. Let $\{e_i\}_{i=1}^{k^4}$ denote the $k^4$ child edges of node $u$. For any child edge $e_i$, let $X_i$ denote the random variable that is 0 if $e_i$ arrives in $I_h$ and 1 otherwise. Define $X = \sum_{i=1}^{k^4} X_{i}$. Since each edge $e_i$ arrives at $r_{e_i} \in [0,1]$ independently and uniformly at random, we have $\mathbb{E}[X_i] = 1/k,\ \forall e_i$. Applying the standard Chernoff bound (Theorem~\ref{thm:chernoff} in Appendix~\ref{sec:chernoff} with $\mu = k^3$ and $\delta = 1 - \frac{1}{k}$), we get 
$\Pr[X < k^2]  =  e^{-k^3 \delta^2 / 2} \leq e^{-k^3 / 8}$ for any $k \geq 3$. 
\end{proof}

\begin{definition}[Bad subtree]
    We say a rooted subtree $T'$ of $T$ is \emph{bad} if it is a full $k^2$-ary tree of height $k$ whose  internal nodes are all bad.
\end{definition}

By definition, the edges of a \emph{bad subtree} arrive in a leaf-to-root order, i.e, for any internal node in the tree, all child edges arrive before the parent edge. We now argue that at least one bad subtree exists with high probability.

\begin{lemma}
\label{lem:badsubtree}
    There exits a bad subtree $T'$ with probability at least $1- \frac{2 k^{2k}}{e^{k^3/8}}$. 
\end{lemma}
\begin{proof}

The root is good with probability at most $\rho := 1 / e^{k^3/ 8}$. Conditioned on the root being bad, consider some arbitrary $k^2$ child edges that all arrive during interval $I_k$ and let $u_1, \ldots u_{k^2}$ be the corresponding child nodes. By Lemma\ref{lem:badnode}, each of those $u_i$'s is good with probability at most $\rho$. So by a union bound, the probability that at least one of those nodes is good is at most $k^2 \rho$.

Using this argument recursively and applying a union bound over all levels, the probability that we have no bad subtree is at most
$\rho (1 + k^2 + \cdots + k^{2k}) \leq 2k^{2k} / e^{k^3/ 8}$.
\end{proof}

\begin{lemma}
    \label{lem:bad-tree-alone}
    Suppose the input graph is a bad tree, with no other edges. In other words, suppose the input graph is a full $k^2$-ary tree of height $k$ where edges in the bottom-to-top order. Then, the maximum load of a node is at least $k$ with a probability at least $1 - 1/ n$.
\end{lemma}
\begin{proof}
    We say a node $u$ of height $h$ is fully loaded if its load is at least $h$. We denote the event as $F_u$. Let $f_h := \Pr[F_u \; | \; F_v \;\; \forall v \in C_u]$ for a node of height $h$; here $C_u$ denotes $u$'s children (nodes). 
        
    We show that $c_h \geq 1-  \frac{1}{2^{k^2 - h}}$. Indeed, consider an arbitrary node $u$ of height $h$. Since we assume that edges arrive in the bottom-to-top order, before any child edges of $u$ arrives, we know that every child node $v$ of $u$ has load at least $h-1$ from the condition $F_v, \forall v \in C_u$. So, the first $h-1$ child edges of $u$ will orient towards $u$. For any other edge $(u, v)$ arriving later, if $v$ already has load more than $h-1$, the claim is immediate. So, suppose not. Then, $u$ can have load $h-1$ only when all the child edges of $u$, except the first $h-1$, are not oriented towards $u$, which occurs with probability $\frac{1}{2^{k^2 - h}} \leq \frac{1}{2^{k^2 - k}}$.

    The bad tree has at most $(k^2)^{k+1}$ internal nodes, and the root has load at least $h$ if $F_u$ occurs for every internal node $u$. Thus, by a simple union bound,
    the root has load at least $k = \Theta(\log n / \log \log n)$ with probability at least $1 - (k^2)^{k+1} \frac{1}{2^{k^2 - k}}  > 1 - 1/ n$. 
\end{proof}

We now explain why we continue to have the lower bound of $k$, independent of the other edges different from the bad tree. This follows from the monotonicity of the greedy algorithm. In other words, inserting an edge (job) to an input sequence can only increase the maximum load. 

Since $k = \Omega(\log n / \log \log n)$, we have the following theorem.

\begin{theorem}
    There exists a tree $T$ with $n$ nodes such that, for uniform random edge arrivals, the greedy algorithm outputs an orientation with maximum load $\Omega(\log n / \log \log n)$ with high probability.
\end{theorem}

\subsection{Lower Bound for Arbitrary Algorithms}
\label{sec:lb-general}

In this section we show a lower bound of $\Omega(\sqrt{\log n})$ for the graph balancing problem with random edge arrivals. Unlike Section~\ref{sec:lb-greedy}, this lower bound is not restricted to the greedy algorithm and holds for all online algorithms.

\vspace{-2mm}
\paragraph{Lower Bound Instance.} We denote the lower bound instance of height $D$ as $T_D$, where $D \geq 0$ is an integer. For every $0 \leq d \leq D$, $T_d$ is constructed recursively as follows: $T_0$ is a singleton node. Let $r$ be the root of $T_d$. For every $d' = 0, 1, \ldots, d-1$, there are $2^{D-d'}$ copies of $T_{d'}$, each of which is a subtree of $r$. Here, it is important to note that $T_1, \ldots, T_D$ have dependency on $D$. See  Figure~\ref{fig:gen-lb-2} for an illustration of the instance $T_2$. 

\vspace{2mm}

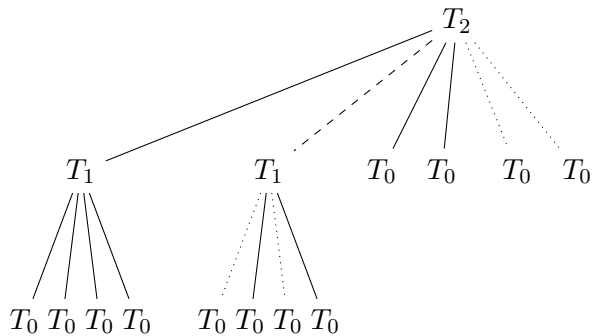
\begin{figure}[htb]
    \centering

\begin{tikzpicture}
  [level 1/.style={sibling distance=20mm}, level distance=10mm,
   level 2/.style={sibling distance=5mm}, level distance=30mm, 
   level 3/.style={sibling distance=5mm}, level distance=20mm]
  \node  (root) {$T_2$}
     child {node {$T_1$} {child {node {$T_0$}} child {node {$T_0$}} child {node {$T_0$}} child  {node {$T_0$}}}} child [dashed,xshift=5mm] {node [black]  {$T_1$} {child [dotted] {node [solid] {$T_0$}} child [solid] {node {$T_0$}} child [dotted] {node [black] {$T_0$}} child [solid] {node {$T_0$}}}}
    child {node {$T_0$}} child [xshift=-12mm] {node {$T_0$}} child [xshift=-22mm, dotted] {node [black] {$T_0$}} child  [xshift=-34mm, dotted]  {node [black] {$T_0$}} ;
   
\end{tikzpicture}
\caption{Example of a lower bound instance for $D=2$. Consider the time when the dashed edge appears. Edges are dotted if they have already arrived, solid otherwise.  For the dashed edge, no algorithm can distinguish which of its end points is the parent node.}
    \label{fig:gen-lb-2}
\end{figure}

\smallskip
We label each node $v$ in $T_D$ with an integer in $[0, D]$: if the sub-tree rooted at $v$ is a copy of $T_d$, then we set $\ell(v) = d$ to be the label of $v$. Therefore, for two integers $0 \leq d' < d \leq D$, any node $u$ with label $d$ has $2^{D - d'}$ children with label $d'$. 

\subsubsection{Intuition}

The main idea behind this lower bound instance is that when some edges arrive online, the algorithm cannot distinguish between their parent node and the child node. Therefore, we can assume that the algorithm assigns such an edge $e = (u, v)$ to $u$ with probability $1/2$ and to $v$ with probability $1/2$. In particular, consider any node $v$ of with label $d$.
Then, for each $d' \in [0, d-1]$, there is at least one edge $(v, u)$ with $\ell(u) =d'$, such that with probability at least $1/2$ the algorithm cannot distinguish the subtree rooted at $u$ consisting of the edges that have arrived so far, and the analogously defined subtree rooted at $v$. Using this observation and the fact that $D = \Omega(\sqrt {\log n})$, we can lower bound the maximum load on the root of the tree. 

\subsubsection{Analysis}

We begin by defining the event when the algorithm cannot distinguish between the parent and child node. 

\begin{definition}
    Let $\pi$ be a random permutation over all edges in $T_D$. Consider an edge $e = (u,v)$ with $u$ being the parent. The permutation $\pi$ is \emph{bad} for $e$, if the following hold:
	\begin{itemize}
		\item the parent edge of $u$ appears after $e$ in the order $\pi$, and
		\item all other child edges $(u,v')$ of $u$ with $\ell(v') \geq \ell(v)$ appear after $(u,v)$ in the order $\pi$. 
	\end{itemize}
\end{definition}

 Consider the arrival time of an edge $e = (u,v)$ where $\ell(u) = d$ and $\ell(v) = d' < d$. We observe that conditioned on the event that the permutation $\pi$ is bad for edge $(u,v)$, the subtrees rooted at $u$ and $v$ comprised only of the edges that have arrived so far are stochastically identical. Thus, we can assume wlog that any fixed algorithm orients the edge $(u,v)$ towards each of its end points with probability $\frac{1}{2}$, which can be formally shown in the following claim. Note that it suffices to consider a fixed deterministic algorithm thanks to Yao's minmax principle \cite{borodin2005online}. 

 \begin{claim}
     Conditioned on the event that the permutation $\pi$ is bad for edge $(u,v)$, any fixed deterministic algorithm orients the edge towards each of its end points with probability $\frac{1}{2}$.
 \end{claim}
 \begin{proof}  
    Let $\psi$ be a random permutation of the node identities that is generated by the adversary, which is unknown to the deterministic algorithm. Let $\pi$ be the random permutation of edges that denotes the arrival order. We condition on the event that the arrival order $\pi$ is bad for some specific edge $(u,v)$. Note that the algorithm makes its decisions based on revealed edges of the form $(\psi(x), \psi(y))$. By our construction, the subtrees $T_u$ and $T_v$ rooted at $u$ and $v$ respectively consisting of arrived edges are stochastically identical. Suppose the algorithm orients the revealed edge $(\psi(u), \psi(v))$ towards $u$ in a specific instantiation. Since $T_u$ and $T_v$ are stochastically identical, they will have their instantiations swapped equiprobably, meaning that the algorithm must orient $(u,v)$ towards $v$ with the same probability (over the random choice of $\psi$ and $\pi$). 
 \end{proof}

\smallskip
In the following lemma, we calculate the probability that a permutation is bad for an edge. For our purpose, it is enough to consider an edge with the root $r$ of $T_D$ as one of its end points.

\begin{lemma}
    A permutation $\pi$ over the edges of $T_D$ is bad for the first appearing edge $(r, u)$ with $\ell(u) = d$ with probability at least $\frac{1}{2}$.
\end{lemma}

\begin{proof}
    The probability can be easily calculated as $$\frac{\text{\#edges with } \ell(v) = d}{\text{\#edges with } \ell(v) \geq d} = \frac{2^{D-d}}{2^{D-d} + \cdots + 2^{D- (D-1)} + 1} \geq \frac{1}{2}. \qedhere$$
\end{proof}

Focus on the root $r$ with $\ell(r) = D$, and an integer $d < D$. Consider the first edge $(r, u)$ with $\ell(u) = d$. With probability 1/2, the permutation is bad for this edge. Conditioned on the permutation being bad, the algorithm assigns the edge to $r$ with probability 1/2. Therefore, the algorithm assigns the edge to $r$ with probability at least $1/4$. Consider all values of $d < D$, in expectation the root will get a load of $\Omega(D)$. The size of the tree is $2^{O(D^2)}$, as we will see below.  Therefore we can obtain a lower bound of $\Omega(\sqrt{\log n})$.

\begin{lemma}
    If $n$ is the size of the tree $T_D$, then $D = \Omega(\sqrt{\log n})$.
\end{lemma}
\begin{proof}
    Let $n(d)$ denote the number of nodes in $T_d$. From the recursive construction of the trees, we have
    $$
    n(d) = \\ 
    \begin{cases}
        1   & d = 0 \\   
        2^{D-0} n(0) +  \cdots + 2^{D -(d-1)} n(d-1) + 1 & d \in [1, D-1]  \\    
    \end{cases}
    $$
    Since $n(d)$ is increasing in $d$, when $d \geq 1$, we have
    $n(d) \leq 2^D D \cdot n(d-1) \leq 4^D \cdot n(d-1)$. Therefore,  $n = n(D) \leq 4^{D^2}$. This gives us the desired result $D = \Omega( \sqrt{\log n})$.
   
\end{proof} 
 
 \begin{theorem}
     There exists a tree $T$ with $n$ nodes such that for uniform random edge arrivals, any algorithm outputs an orientation with maximum load $\Omega(\sqrt{\log n})$, with high probability.
 \end{theorem}

\section{Concluding Remarks}

In this paper, we revisit the classic online load balancing problem for unrelated machines in the random arrival order model. While we achieve an exponential improvement over the previous lower bound, a substantial gap remains  between the lower bound of $\Omega(\sqrt {\log m} )$ and the upper bound of $O(\log m / \log \log m)$. Importantly, the question of whether an algorithm with a better competitive ratio exists even for tree inputs in the online graph balancing problem remains open. This holds true even when the entire tree structure is known a priori (with the identities of nodes and edges hidden).

\section{Acknowledgements}
    Petety is supported in part by NSF grants CCF-1844939 and CCF-2121745. Im is supported in part by an ONR grant N00014-22-1-2701 as well as them. Li is supported in part by the State Key Laboratory for Novel Software Technology, and the New Cornerstone Science Laboratory. 



\begin{thebibliography}{10}

\bibitem{Albers99}
Susanne Albers.
\newblock Better bounds for online scheduling.
\newblock {\em {SIAM} J. Comput.}, 29(2):459--473, 1999.

\bibitem{albers2021scheduling}
Susanne Albers and Maximilian Janke.
\newblock Scheduling in the random-order model.
\newblock {\em Algorithmica}, 83(9):2803--2832, 2021.

\bibitem{aspnes1997line}
James Aspnes, Yossi Azar, Amos Fiat, Serge Plotkin, and Orli Waarts.
\newblock On-line routing of virtual circuits with applications to load balancing and machine scheduling.
\newblock {\em JACM}, 44(3):486--504, 1997.

\bibitem{awerbuch1995load}
Baruch Awerbuch, Yossi Azar, Edward~F Grove, Ming-Yang Kao, P~Krishnan, and Jeffrey~Scott Vitter.
\newblock Load balancing in the {$L_p$} norm.
\newblock In {\em FOCS}, pages 383--391, 1995.

\bibitem{azar2005line}
Yossi Azar.
\newblock On-line load balancing.
\newblock {\em Online Algorithms: The State of the Art}, pages 178--195, 2005.

\bibitem{azar2013tight}
Yossi Azar, Ilan~Reuven Cohen, Seny Kamara, and Bruce Shepherd.
\newblock Tight bounds for online vector bin packing.
\newblock In {\em STOC}, pages 961--970, 2013.

\bibitem{azar1995competitiveness}
Yossi Azar, Joseph Naor, and Raphael Rom.
\newblock The competitiveness of on-line assignments.
\newblock {\em Journal of Algorithms}, 18(2):221--237, 1995.

\bibitem{berman2000line}
Piotr Berman, Moses Charikar, and Marek Karpinski.
\newblock On-line load balancing for related machines.
\newblock {\em Journal of Algorithms}, 35(1):108--121, 2000.

\bibitem{plank2017online}
Plank B.M.
\newblock Online scheduling problems in the random order model.
\newblock {\em Bachelor's Thesis}, 2017.

\bibitem{borodin2005online}
Allan Borodin and Ran El-Yaniv.
\newblock {\em Online Computation and Competitive Analysis}.
\newblock Cambridge University Press, 2005.

\bibitem{buchbinder2006fair}
Niv Buchbinder and Joseph Naor.
\newblock Fair online load balancing.
\newblock In {\em SPAA}, pages 291--298, 2006.

\bibitem{caragiannis2008better}
Ioannis Caragiannis.
\newblock Better bounds for online load balancing on unrelated machines.
\newblock In {\em SODA}, pages 972--981, 2008.

\bibitem{ebenlendr2008graph}
Tom{\'a}s Ebenlendr, Marek Krc{\'a}l, and Jiri Sgall.
\newblock Graph balancing: a special case of scheduling unrelated parallel machines.
\newblock In {\em SODA}, pages 483--490, 2008.

\bibitem{ebenlendr2014graph}
Tom{\'a}{\v{s}} Ebenlendr, Marek Kr{\v{c}}{\'a}l, and Ji{\v{r}}{\'\i} Sgall.
\newblock Graph balancing: A special case of scheduling unrelated parallel machines.
\newblock {\em Algorithmica}, 68:62--80, 2014.

\bibitem{ghomi2017load}
Einollah~Jafarnejad Ghomi, Amir~Masoud Rahmani, and Nooruldeen~Nasih Qader.
\newblock Load-balancing algorithms in cloud computing: A survey.
\newblock {\em Journal of Network and Computer Applications}, 88:50--71, 2017.

\bibitem{Graham66}
R.~L. Graham.
\newblock Bounds for certain multiprocessing anomalies.
\newblock {\em SIAM J. Appl. Math.}, 1966.

\bibitem{Graham69}
R.~L. Graham.
\newblock Bounds on multiprocessing timing anomalies.
\newblock {\em SIAM J. Appl. Math.}, 17:416--429, 1969.

\bibitem{9719768}
Anupam Gupta, Gregory Kehne, and Roie Levin.
\newblock Random order online set cover is as easy as offline.
\newblock In {\em FOCS}, pages 1253--1264, 2022.

\bibitem{gupta2014experts}
Anupam Gupta and Marco Molinaro.
\newblock How experts can solve {LP}s online.
\newblock In {\em ESA}, pages 517--529, 2014.

\bibitem{gupta_singla_2021}
Anupam Gupta and Sahil Singla.
\newblock {\em Random-Order Models}, page 234–258.
\newblock Cambridge University Press, 2021.

\bibitem{ImKKP19}
Sungjin Im, Nathaniel Kell, Janardhan Kulkarni, and Debmalya Panigrahi.
\newblock Tight bounds for online vector scheduling.
\newblock {\em {SIAM} J. Comput.}, 48(1):93--121, 2019.

\bibitem{jiang2015survey}
Yichuan Jiang.
\newblock A survey of task allocation and load balancing in distributed systems.
\newblock {\em IEEE TPDS}, 27(2):585--599, 2015.

\bibitem{kaplan2023almost}
Haim Kaplan, David Naori, and Danny Raz.
\newblock Almost tight bounds for online facility location in the random-order model.
\newblock In {\em SODA}, pages 1523--1544, 2023.

\bibitem{karger2004simple}
David~R Karger and Matthias Ruhl.
\newblock Simple efficient load balancing algorithms for peer-to-peer systems.
\newblock In {\em SPAA}, pages 36--43, 2004.

\bibitem{LattanziLMV20}
Silvio Lattanzi, Thomas Lavastida, Benjamin Moseley, and Sergei Vassilvitskii.
\newblock Online scheduling via learned weights.
\newblock In {\em SODA}, pages 1859--1877, 2020.

\bibitem{lenstra1990approximation}
Jan~Karel Lenstra, David~B Shmoys, and {\'E}va Tardos.
\newblock Approximation algorithms for scheduling unrelated parallel machines.
\newblock {\em Math. Prog.}, 46:259--271, 1990.

\bibitem{li2021online}
Shi Li and Jiayi Xian.
\newblock Online unrelated machine load balancing with predictions revisited.
\newblock In {\em ICML}, pages 6523--6532, 2021.

\bibitem{meyerson2013online}
Adam Meyerson, Alan Roytman, and Brian Tagiku.
\newblock Online multidimensional load balancing.
\newblock In {\em APPROX-RANDOM}, pages 287--302, 2013.

\bibitem{molinaro2017online}
Marco Molinaro.
\newblock Online and random-order load balancing simultaneously.
\newblock In {\em SODA}, pages 1638--1650, 2017.

\end{thebibliography}

\appendix
\section{Chernoff Bound}
    \label{sec:chernoff}

\begin{theorem}
    \label{thm:chernoff}
    Let $X_{1}, \dots , X_{n}$ be $n$ independent, Bernoulli random variables, where $\mathbb{P}[X_{i} = 1] = p_{i}$, for all $i$. For $X = \sum X_{i}$, its expectation is $\mu = \E[X] = \sum_{i}p_i$. Then for any $0 < \delta < 1$, we have
    \begin{equation}
        P[X < (1-\delta) \mu ] < \left(\frac{e^{-\delta}}{(1-\delta)^{1-\delta}}\right)^\mu.
    \end{equation}
Further, the upper bound is at most $e^{-\mu \delta^2 / 2}$ when $\delta \in [0, 1)$.
\end{theorem}

\end{document}